\documentclass[a4paper,10pt]{amsart}
\usepackage[T1]{fontenc}
\usepackage[latin1]{inputenc}
\usepackage{amssymb,amsmath,amsthm,amsrefs}

\newtheorem{definition}{Definition}
\newtheorem{theorem}{Theorem}
\newtheorem{lemma}{Lemma}

\newtheorem*{assp*}{\textbf{(CM$\nu$) Conditional probability distribution}}
\newtheorem*{assi*}{\textbf{(I) Short-range interaction}}
\newtheorem*{dsknN*}{(\textbf{DS}.$k,n,N$)}
\numberwithin{equation}{section}
\numberwithin{theorem}{section}
\numberwithin{definition}{section}
\numberwithin{lemma}{section}
\DeclareMathOperator{\dist}{dist}
\DeclareMathOperator{\diam}{diam}
\DeclareMathOperator{\supp}{supp}
\DeclareMathOperator{\prob}{\mathbb{P}}
\newcommand{\condI}{\mathbf{(I)}}

\newcommand{\cmu}{\textbf{(CM$\nu$)}}
\newcommand{\ee}{\mathrm{e}}
\newcommand{\esssup}{\mathrm{ess\ sup}}
\newcommand{\card}{\mathrm{card}}
\newcommand{\Bone}{\mathbf{1}}

\newcommand{\BH}{\mathbf{H}}
\newcommand{\BDelta}{\mathbf{\Delta}}

\newcommand{\dsonN}{(\textbf{DS}$0,n,N$)}

\newcommand{\dsknN}{(\textbf{DS}$.k,n,N$)}
\newcommand{\dskonN}{(\textbf{DS}$.k+1,n,N$)}
\newcommand{\dsknprimeN}{(\textbf{DS}$.k,n',N$)}

\newcommand{\dskprimenN}{(\textbf{DS}$.k',n,N$)}
\newcommand{\BPsi}{\mathbf{\Psi}}

\newcommand{\DZ}{\mathbb{Z}}
\newcommand{\DR}{\mathbb{R}}
\newcommand{\DN}{\mathbb{N}}
\newcommand{\esm}{\mathbb{E}}
\newcommand{\DP}{\mathbb{P}}
\newcommand{\DC}{\mathbb{C}}
\newcommand{\DD}{\mathbb{D}}

\newcommand{\BC}{\mathbf{C}}
\newcommand{\BX}{\mathbf{X}}
\newcommand{\BK}{\mathbf{K}}
\newcommand{\BP}{\mathbf{P}}
\newcommand{\BG}{\mathbf{G}}

\newcommand{\BU}{\mathbf{U}}
\newcommand{\BV}{\mathbf{V}}

\newcommand{\Bx}{\mathbf{x}}
\newcommand{\By}{\mathbf{y}}
\newcommand{\Bu}{\mathbf{u}}
\newcommand{\Bv}{\mathbf{v}}
\newcommand{\Bw}{\mathbf{w}}
\newcommand{\Bz}{\mathbf{z}}
\newcommand{\FB}{\mathfrak{B}}
\newcommand{\CA}{\mathcal{A}}
\newcommand{\CR}{\mathcal{R}}

\newcommand{\CE}{\mathcal{E}}

\title[Resonances  and low energy localization]{ Resonances and multi-particle localization  at low energy}
\author{Tr\'esor Ekanga}
\address{Universit\'e Paris Diderot 13 Rue Albert Einstein 75013 Paris France and  Universit\'e de Yaound\'e I Cameroun}
\email{ekanga@math.cnrs.fr}
\begin{document}
\begin{abstract}
We use a new eigenvalue concentration bound for the fluctuation of the sample mean of the random extternal potential in the multi-particle Anderson model and prove the spectral exponential and the strong dynamical localization. The results just need some weaker condition on the distribution function of the sample including a large class of probability distribution such as log-H\"older continuous or Lipshitz continuous. The method also apply to i.i.d. random Gaussian  potential with independent fluctuations.
\end{abstract}
\maketitle

\section{Introduction and the results}

In this paper we consider the multi-particle discrete Anderson model with Gaussian random external potentials and using  a result on concentration of eigenvalues \cite{Chu10a}, we prove the spectral exponential and the strong dynamical localization at low energy.

The novelty of the paper is that  the results are applicable to the Gaussian sample mean with independent fluctuations as in the paper \cite{Chu10b} where localization has been obtained in the high disorder limit. This paper under the low energy regime  complements  that work.

In some numerous previous works in this direction \cites{AM93,AW09,AW10,DK89,CS08,CS09a,CS09b}, the authors do not treat the case of the sample mean with fluctuations such as the Gaussian random potentials.

In our earlier works \cites{Eka11,Eka19a,Eka19b,Eka20}, we just consider the classical standard Anderson model with either i.i.d. or correlated random external potential without the consideration of the Gaussian sample mean which remains a more complicated situation in the mathematics of disordered quantum systems.

For the validity of the new eigenvalue concentration bound, we need a very general  hypothesis on the conditional probability of the fluctuation  of the sample mean. This last assumption is much weaker  than H\"older or even log-H\"older continuity and cover a larger class of discrete random Anderson models than  in the papers \cites{FMSS85,GB98,Kir08,Sto01}.

The proof of our localization results use a different form of the multi-scale analysis following the scheme developed in the paper \cite{Chu10b} and extend it to the low energy regime. In the case of random i.i.d. random potential in the Anderson model with $0$ mean and unit variance, the sample mean is Gaussian  random variable with a bounded probability density.

Let us now discuss on the structure of the paper. Our main results  are Theorem \ref{thm:main.result.exp.loc} the exponential spectral localization  and Theorem \ref{thm:main.result.dynamical.loc} the strong dynamical localization at low energy. The next Section, Section \ref{sec:MSA.scheme} is devoted  to the multi-scale analysis scheme. In Section \ref{MP:induction} we present  the multi-particle multi-scale induction step while in Sections \ref{sec:proofs.results} we prove the localization results.

\subsection{The model}
We fix the number of particles $N\geq 2$. We are concern  with multi-particle random  Schr\"odinger operators of the following form:
\[
\BH^{(N)}(\omega):=-\BDelta + \BU+\BV,
\]
acting in $\ell^2((\DZ^d)^N)$. Sometimes, we will use the identification $(\DZ^d)^N\cong \DZ^{Nd}$. Above, $\BDelta$ is the Laplacian on $\DZ^{Nd}$, $\BU$ represents the inter-particle interaction which acts as multiplication operator in $\ell^2(\DZ^{Nd})$. Additional information on $\BU$ is given in the assumptions. $\BV$ is multi-particle random external potential also acting  as multiplication operator on ${\ell}^2(\DZ^{Nd})$. For $\Bx=(x_1,\ldots,x_N)\in (\DR^d)^N$, $\BV(\Bx)=V(x_1)+\cdots+V(x_N)$ and $\{V(x,\omega), x \in\DZ^d\}$ is a random stochastic process relative to some probability space $(\Omega,\FB,\DP)$.

Observe  that the non-interacting Hamiltonian $\BH_0^{(N)}(\omega)$ can be written  as a tensor  product:
\[
\BH_0^{(N)}(\omega):=-\BDelta + \BV =\sum_{k=1}^N \Bone_{{\ell}^2(\DZ^d)}^{\otimes^{(k-1)}}\otimes H^{(1)}(\omega)\otimes \Bone_{\ell^2(\DZ^d)}^{\otimes^{(N-k)}},
\]
where, $H^{(1)}(\omega)=-\Delta +V(x,\omega)$ acting on $\ell^2(\DZ^d)$. We will also consider random Hamiltoninans $\BH^{(n)}(\omega)$, $n=1,\ldots,N$ defined similarly. Denote by $|\cdot|$ the max-norm in $\DR^{nd}$
Denote by $\DD$ the principal diagonal on $(\DZ^d)^N$ 
\[
\DD=\{\Bx\in(\DZ^d)^N: \Bx=(x,\ldots,x), x\in\DZ^d\}.
\]
Using the symmetry of the random potential $\BV(\Bx,\omega)=V(x_1,\omega)+\cdots+V(x_N,\omega)$ we can define the symmetrize distance 
\[
d_S(\Bx,\By):=\min_{\tau\in \mathfrak{S}_{N}}|\Bx-\tau\By|.
\]

\subsection{Some basic geometry and the assumptions}

The multi-particle multi-scale analysis required the consideration of lattice cubes: for  $\Bu\in\DZ^{nd}$ with coordinates $\{u_1,\ldots,u_n\}\subset\DZ^d$ and given $L\in(0,\infty)$ set,

\[
C^{(1)}_{L_i}(u_i)=\left\{x\in\DZ^d: |x-u_i|\leq L_i\right\}.
\]
By $\BC^{(n)}_L(\Bu)$ we denote the $n$-particle cube, i.e.,
\[
\BC^{(n)}_L(\Bu)=\left\{\Bx\in\DZ^{nd}: |\Bx-\Bu |\leq L\right\}
\]
We define the \emph{boundary} of the domain $\BC_L^{(n)}(\Bu)$ by 
\[
\partial\BC^{(n)}(\Bu)=\{ (\Bv,\Bv')\in\DZ^{nd}\times\DZ^{nd}\mid |\Bv-\Bv'|_1=1 \text{ and}
\]
\[
 \text{either $\Bv\in\BC^{(n)}(\Bu)$, $\Bv'\notin\BC^{(n)}(\Bu)$ or $\Bv\notin\BC^{(n)}(\Bu)$, $\Bv'\in\BC^{(n)}(\Bu)$}\} 
\]
its \emph{internal boundary} by 
\begin{equation}
\partial^-\BC^{(n)}(\Bu)=\left\{\Bv\in\DZ^{nd}: \dist\left(\Bv,\DZ^{nd}\setminus\BC^{(n)}(\Bu)\right)=1\right\}
\end{equation}
and its \emph{external boundary} by 
\begin{equation}
\partial^+\BC^{(n)}(\Bu)=\left\{\Bv\in\DZ^{nd}\setminus \BC^{(n)}(\Bu): \dist\left(\Bv,\BC^{(n)}(\Bu)\right)=1\right\}
\end{equation}

\begin{assi*}
The potential of inter-particle interaction
\[
\BU: (\DZ^d)^n\rightarrow\DR
 \]
is bounded and of the form
\[
\BU(\Bx)= \sum_{1\leq i\leq j\leq n \quad i\neq j} \Phi(|x_i-x_j|),
\]
where the points $\{x_i,i=1\ldots,n\}$ represent the coordinates of $\Bx\in(\DZ^d)^n$ and $\Phi:\DN\rightarrow\DR_+$ is a compactly supported non-negative function:
\[
\exists r_0\in\DN: \quad \supp\Phi\subset[0,r_0].
\]
\end{assi*}
 $r_0$ refers as the ''range'' of the interaction $\BU$.

We will make use of the following notations. Given a parallelepiped $Q\subset\DZ^d$, we denote by $\xi_Q(\omega)$ the sample mean of the random field $V$ over $Q$.
\[
\xi_Q(\omega)=\frac{1}{|Q|} \sum_{x\in Q} V(x,\omega),
\]
and consider the fluctuations of the random potential $V$ relative to the sample mean 
\[
\eta_{x}= V(x,\omega)-\xi_Q(\omega),\quad {x\in Q}
\]
and denote by $\mathfrak{F}_{V,Q}$ the sigma algebra generated by 
\[
\{ \eta_x, x\in Q: V(y,\cdot), y\notin Q\}
\] 
and by $F_{\xi}(\cdot|\mathfrak{F}_{V,Q})$ the conditional distribution function of $\xi_Q$ given $\mathfrak{F}_{V,Q}$
\[
F_{\xi}(s|\mathfrak{F}_{V,Q}):=\prob\{\xi_Q\leq s|\mathfrak{F}_{V, Q}\}.
\]
The random potential is assumed to satisfy the following condition: 
\begin{assp*}
For any $R\geq 0$ there exists a function $\nu_R: \DR_+\rightarrow\DR_+$ which equals zero at $0$ and such that $\forall Q\subset\DZ^{d}$ with $\diam(Q)\leq R$, we have that:
\[
\forall t,s\in\DR, \quad \esssup |F_{\xi}(t|\mathfrak{F}_{V,Q})-F_{\xi}(s|\mathfrak{F}_{V,Q})|\leq \nu_{R}(|t-s|)
\]
This condition  is useless if $\nu_R(s)\rightarrow 0$ too slowly as $s\rightarrow 0$. For that reason we make use of the following  condition 
\[
\nu_R(t)\leq Const R^A \ln^{-B}|t|^{-s},\quad  |t|\in(0,1),
\]
with $A\in(0,\infty)$ and sufficiently large positive $B$
\end{assp*}

\subsection{The main results}

\begin{theorem}\label{thm:main.result.exp.loc}
Under the assumptions $\condI$ and $\cmu$ there exists  $E^*$ bigger than $E_0^{(n)}$ such that with $\DP$-probability $1$:
\begin{enumerate}
\item[(i)] the spectrum of $\BH^{(N)}(\omega)$ in $[E_0^{(N)},E^*]$ is pure point,\\
\item[(ii)] any eigenfunction $\BPsi_i(\Bx,\omega)$ with eigenvalue  $E_i(\omega)\in[E_0^{(N)},E^*]$ is exponentially decaying fast at infinity: There exists a non-random constant $m$ and a positive random constant $C_i(\omega)$ such that
\[
|\BPsi_i(\Bx,\omega)|\leq C_i(\omega)\ee^{-m|\Bx|}
\]
\end{enumerate}
\end{theorem} 

\begin{theorem}\label{thm:main.result.dynamical.loc}
Under the assumptions $\condI$ and $\cmu$, there exists $E^*$ bigger than $E_0^{(n)}$ and a positive $s^*$ such that for any  bounded domain $\BK\subset\DZ^{Nd}$ and any $s\in(0,s^*)$ we have that the following quantity:
\[
\esm\left[\sup_{\|f\|_{\infty}\leq 1}\left\| |\BX^{s/2}f(\BH^{(N)}(\omega))\BP_{I}(\BH^{(N)}(\omega))\Bone_{\BK}\right\|\right]
\]
is finite where $(|\BX|\BPsi)(\Bx):=[\Bx|\BPsi(\Bx)$, $\BP_I(\BH^{(N)}(\omega))$ is the spectral projection  of $\BH^{(N)}(\omega)$ onto the interval $I=[E^{(N)}_0,E^*]$ and the supremum is taken over bounded measurable functions $f$.
\end{theorem}
\section{The multi-particle multi-scale analysis scheme}\label{sec:MSA.scheme}

 We define the restriction of the Hamiltonian $\BH^{(n)}$ to $\BC^{(n)}(\Bu)$ by
$\BH^{(n)}_{\BC^{(n)}(\Bu)}=\BH^{(n)}$ with simple boundary conditions on $\partial^+\BC^{(n)}(\Bu)$ i.e., $\BH^{(n)}_{\BC^{(n)}(\Bu)}(\Bx,\By)=\BH^{(n)}(\Bx,\By)$ whenever  $\Bx,\By\in\BC^{(n)}(\Bu)$ and $\BH^{(n)}_{\BC^{(n)}(\Bu)}(\Bx,\By)=0$ otherwise. We denote the spectrum of $\BH^{(n)}_{\BC^{(n)}(\Bu)}$ by $\sigma(\BH^{(n)}_{\BC^{(n)}(\Bu)})$ and its resolvent by 

\begin{equation}
\BG_{\BC^{(n)}(\Bu)}(E):=\left(\BH^{(n)}_{\BC^{(n)}(\Bu)}-E\right)^{-1}, \quad E\in\DR\setminus\sigma\left(\BH^{(n)}_{\BC^{(n)}(\Bu)}\right).
\end{equation}

The matrix elements $\BG_{\BC^{(n)}(\Bu)}(\Bx,\By;E)$ are usually called the \emph{Green functions} of the operator $\BH^{(n)}_{\BC^{(n)}(\Bu)}$. The multi-scale analysis is based on a length scale $\{L_k\}_{k\geq 0}$ which is chosen as follows.

\begin{definition}
The length scale $\{L_k\}_{k\geq 0}$ is a sequence of integers defined by the initial length scale $L_0\geq 3$ and by the recurrence relation 
\[
L_{k+1}=\lfloor L_k^{\alpha}\rfloor +1
\]
where $\alpha\in(1,2)$ is some fixed number. In this paper, $\alpha=3/2$.
\end{definition}

\begin{definition}
Given $E\in\DR$ and a positive $m$, a cube $\BC^{(n)}_L(\Bu)$ is said $(E,m)$-non-singular ($(E,m)$-NS ) if 
\begin{equation}
\max_{\Bx: |\Bx-\Bu|\leq L^{1/\alpha}}\max_{\By\in\partial^-\BC^{(n)}_L(\Bu)} |(\BH^{(n)}_{\BC^{(n)}_L(\Bu)}-E)^{-1}(\Bx,\By)|\leq \ee^{-\gamma(m,L,n)L}
\end{equation}
where $\gamma(m,L,n)=m(1+L^{-1/4})^{N-n+1}$.
Otherwise it is called $(E,m)$-singular ($(E,m)$S).
\end{definition}

Now we are ready to present our key bound on the multi-particle multi-scale analysis. The bound is obtain by an induction procedure in both the traditional length scale and alsoon the numder of particle for a fixed number of particle $N$ and fixed parameter $m$ in the decay of the Green functions. We summarize it the following:

\begin{dsknN*}
For any pair of $2nL_k$-distant pairs of $n$-particle cubes $\BC^{(n)}_{L_k}(\Bu)$, $\BC^{(n)}_{L_k}(\Bv)$ we have that
\begin{equation}\label{eq:MSA}
\prob\left\{ \exists E\in I: \text{$ \BC^{(n)}_{L_k}(\Bu)$ and $\BC^{(n)}_{L_k}(\Bv)$ are $(E,m)$-S}\right\}\leq L_k^{-p2^{N-n+1}}
\end{equation}
\end{dsknN*}

\subsection{Eigenvalue concentration bound} 

\begin{definition}
Let $n\geq 1$, $\beta=1/2$  and $E\in\DR$ be given. Consider a rectangle $\BC^{(n)}(\Bu)=\prod_{i=1}^n C^{(1)}_{L_i}(u_i)$ and set $L=\min_{i=1,\ldots,n}\{L_i\}$. $\BC^{(n)}(\Bu)$ is called  $E$-resonant ($E$-R) if 
\[
\dist\left[E,\sigma(\BH^{(n)}_{\BC^{(n)}(\Bu)})\right]\leq \ee^{-L^{\beta}}.
\]
Otherwise it is called $E$-non-resonant ($E$-NR)
\end{definition}

\begin{theorem}\label{thm:Wegner}[Cf. \cite{Chu10a}]
Assume that the random potential satisfies assumption $\cmu$, then for any pair of $n$-particle operators $n=1,\ldots,N$ $\BH^{(n)}_{\BC^{(n)}_{L'} (\Bu)}$ and $\BH^{(n)}_{L''}(\Bv)$ with $0\leq L',L''\leq L$ satisfying $d_S(\Bu,\Bv)\geq 2NL$ and any $s\in(0,\infty$ we have that 
\begin{equation}\label{eq:Wegner}
\prob\left\{ \dist(\sigma(\BH^{(n)}_{\BC^{(n)}_{L'}(\Bu)}),\sigma(\BH_{\BC^{(n)}_{L''}(\Bv)}))\leq s\right\}\leq |\BC^{(n)}_{L'}(\Bu)|\cdot|\BC^{(n)}_{L''}(\Bv)|\cdot\nu_R(2s)
\end{equation}
\end{theorem}

\subsection{Initial length scales estimates of the multi-scale analysis}

\begin{lemma}\label{lem:CT}[Cf. \cite{Kir08}]  
Consider  a lattice Schr\"odinger operator 
\[
H_{\Lambda}=-\Delta_{\Lambda}+ W(x)
\]
acting in $\ell^2(\Lambda)$, $\Lambda\subset\DZ^{D}$, $D\geq 1$ with an arbitrary\footnote{
This includes the cases of single-particle and multi-particle operators, that differ only by their potentials.} potential $W:\Lambda\rightarrow\DR$. Suppose  that  $E\in\DR$ is such that\footnote{ Theorem 11.2  from \cite{Kir08} is formulated with the equality $\dist(E,\sigma)=\eta$, but it is clear from the proof  that it remains  valid if $\dist(E,\sigma(H_{\Lambda}))\geq \eta$} $\dist(E,\sigma(H_{\Lambda}))\geq \eta$  with $\eta\in(0,1]$. Then 
\begin{equation}
\forall x,y\in\Lambda \quad \left|(H_{\Lambda}-E)^{-1}(x,y)\right|\leq 2\eta^{-1}\ee^{-\frac{\eta}{12D}|x-y|}
\end{equation}
\end{lemma} 

\begin{theorem}\label{thm:E0}[Cf. \cite{Eka20}]
Let $1\leq n\leq N$. Under assumptions $\condI$ and $\cmu$, we have with probability one:
\[
[0,4nd]\subset  \sigma(\BH^{(n)}(\omega))\subset[0,+\infty).
\]
Consequently
\[
E_0^{(n)}:\inf\sigma(\BH^{(n)}(\omega))=0 \quad \text{a.s.}
\]
\end{theorem}

\begin{lemma}\label{lem:prob.E0.1}[Cf. \cite{Eka20}]
Let $H^{(1)}(\omega)=-\BDelta+V(x,\omega)$  be a random single-particle lattice Schr\"odinger operator in $\ell^2(\DZ^d)$. Assume that the random variables $V(x,\omega)$ are i.i.d. and non-negative. Then, for any positive $C$ there exist  arbitrary large $L_0(C)$ and positive  constants $C_1,c$  such that for any cube $C^{(1)}_{L_0}(u)$, the lowest eigenvalue $E_0^{(1)}(\Omega)$ of $H^{(1)}_{C^{(1)}_{L_0}(u)}(\omega)$ satisfies 
\begin{equation}
\prob\{E^{(1)}_0(\omega)\leq 2 CL_0^{-1/2}\}\leq C_1 L_0^d\ee^{-cL_0^{d/4}}.
\end{equation}
\end{lemma}
\begin{proof}
It follows from equation (11.23) from the proof of Theorem 11.4 in \cite{Kir08}. The absolute continuity of the distribution of the random variables is not required for this result, so it applies to our model. Lemma \ref{lem:prob.E0.1} follows actually from the study of Lifshitz tails and is based on a large deviation estimate valid for i.i.d. processes see Lemma 6.4 in\cite{Kir08} and Theorem 2.1.3 in \cite{Sto01}
\end{proof} 

\begin{lemma}\label{lem:prob.E0.2}[Cf. \cite{Eka20}]
Let $\BH^{(n)}(\omega)=-\BDelta +V(x_1,\omega)+\cdots+V(x_n,\omega)+\BU(\Bx)$ be an $n$-particle random Schr\"odinger operator in $\ell^2(\DZ^{nd})$ where $\BU$ and $V$ satisfy $\condI$ and $^\cmu$ respectively. Then, for any positive $C$ there exist  arbitrary large $L_0(C)$ and positive constants $C_1,c$ such thatfor any cube  $\BC^{(n)}_{L_0}(\Bu)$ the lowest eigenvalue $E_0^{(n)}(\omega)$ of $\BH^{(n)}_{\BC^{(n)}_{L_0}(\Bu)}(\omega)$ satisfies 
\begin{equation}
\prob\{E_0^{(n)}(\omega)\leq 2 CL_0^{-1/2}\}\leq C_1L_0\ee^{-cL_0^{1/4}}.
\end{equation}
\end{lemma} 

\begin{theorem}\label{thm:ILS}[initial scale estimate]. Under assumptions $\condI$ and $\cmu$, for any positive $p$ there exists arbitrarily large $L_0(N,d,p)$ such that if $m=(14N^N+6Nd)L0^{-1/2}$ and $E^*:=(12Nd)(2^{N+1}m)$, then $\dsonN$ is valid  for  any $n=1,\ldots,N$.
\end{theorem}
\begin{proof}
Set $C=12Nd\cdot 2^{N+1}[14N^N+6Nd]$ and let $\BC^{(n)}_{L_0}(\Bu)$ be a cube in $\DZ^{nd}$. Consider $\omega\in\Omega$ such that the first eigenvalue   $E_0^{(n)}(\omega)$ of $\BH^{(n)}_{\BC^{(n)}_{L_0}(\Bu)}(\omega)$ satisfies $E_0^{(n)}(\omega)\geq 2 CL_0^{-1/2}$. Then for all $E\leq CL_0^{-1/2}=E^*$ we have 
\[
\dist(E,\sigma(\BH^{(n)}_{\BC^{(n)}_{L_0}(\Bu)}(\omega)))= E_0^{(n)}(\omega)-E\geq CL_0^{-1/2}=:\eta.
\]
For $L_0$ large enough $\eta\leq 1$ and the Combes-Thomas estimate (Lemma \ref{lem:CT}) implies that for any $\Bv\in\partial^-\BC^{(n)}_{L_0}(\Bu)$:
\[
|\BG^{\BC^{(n)}_{L_0}(\Bu)}(E,\Bu,\Bv)|\leq 2\eta^{-1}\ee^{-\frac{CL_0^{-1/2}}{12Nd}L_0}.
\]
Now observe that $\frac{CL_0^{-1/2}}{12Nd}=2^{N+1}m$. Thus 
\begin{align*}
|\BG_{\BC^{(n)}_{L_0}(\Bu)}(E,\Bu,\Bv)|&\leq 2C^{-1}L_0^{1/2}\ee^{-2^{N+1}mL_0}\\
&\leq 2C^{-1}L_0^{1/2}\ee^{-2\gamma(m,L_0,n)L_0}\\
&\leq \ee^{-\gamma(m,L_0,n)L_0}
\end{align*}
for $L_0$ large enough, since 
\[
\gamma(m,L_0,n)=m(1+L^{-1/8})^{N-n+1}\leq m\times 2^N.
\]
This implies that $\BC^{(n)}_{L_0}(\Bu)$ is $(E,m)$-NS. Thus \text{$\omega\in \{\omega\in\Omega\mid\BC^{(n)}_{L_0}(\Bu)$ is $(E,m)$-NS $\forall E\leq E^*\}$}. Therefore 
\[
\prob\left\{\exists E\leq E^*, \text{$\BC^{(n)}_{L_0}(\Bu)$ is $(E,m)$-S}\right\}\leq \prob\left\{E^{(n)}_0(\omega)\leq 2 CL_0^{-1/2}\right\},
\]
and by Lemma \ref{lem:prob.E0.2}
\[
\prob\{E_0^{(n)}(\omega)\leq 2C L_0^{-1/2}\}\leq C_1 L_0^d\ee^{-cL_0^{1/4}}.
\]
Finally, the quantity $C_1L_0^d\ee^{-cL_0^{1/4}}$ for $L_0$ sufficiently large  is less than$L_0^{-2p4^{N-n}}$. This proves the claim since the probability of two cubes to be singular at the same energy is bounded by the probability of either one of them to be singular.
\end{proof}
In the rest of the paper, the parameters $m$ and $E^*$ are as in Theorem \ref{thm:ILS}.

\section{Multi-particle multi-scale induction}\label{MP:induction}

We begin by presenting the key well-known facts concerning the induction step of the multi-scale analysis namely the geometric resolvent identity and inequality. The bounds rely to the property of non-singularity between a given cube and its sub-cubes.
Let be given a cube $\BC^{(n)}_L(\Bu)$ and sub-cubes $\BC^{(n)}_{\ell}(\Bw)\subset\BC^{(n)}_L(\Bu)$ with $L\geq \ell+1$ for any complex $\zeta$ which is not an eigenvalue of $\BH^{(n)}_{\BC^{(n)}_L(\Bu)}$ $\BH^{(n)}_{\BC^{(n)}_{\ell}(\Bw)}$ and $\Bx\in\BC^{(n)}_{\ell}(\Bw)$ $\By\in \BC^{(n)}_L(\Bu)\setminus\BC^{(n)}_{\ell}(\Bw)$:
\begin{equation}\label{eq:GRE}
\BG^{(n)}_{\BC^{(n)}_L(\Bw)}(\Bx,\By,\zeta)=\sum_{\Bv,\Bv'\in\partial^-\BC^{(n)}_L(\Bw)}\BG^{(n)}_{\BC^{(n)}_{\ell}(\Bw)}(\Bx,\Bv,\zeta)\BG^{(n)}_{\BC^{(n)}_L(\Bu)}(\Bv',\By,\zeta).
\end{equation}
Now, we can deduce the Geometric resolvent inequality  
\begin{align}
&|\BG^{(n)}_{\BC^{(n)}_L(\Bu)}\Bx,\By,\zeta)|\leq |\partial^-\BC^{(n)}_{\ell}(\Bw)| \max_{\Bv,\Bv'}|\BG^{(n)}_{\BC^{(n)}_{\ell}(\Bw)}(\Bx,\Bv,\zeta)|\\
&\cdot|\BG^{(n)}_{\BC^{(n)}_L(\Bu)}(\Bv',\By,\zeta)|\\
&\leq |\partial^-{\BC^{(n)}_{\ell}(\Bw)}\left(\max_{\Bv\in \BC^{(n)}_{\ell}(\Bw)}|\BG^{(n)}_{\BC^{(n)}_{(\ell)}(\Bw)}(\Bx,\Bv,\zeta)|\right)\\
&\max_{\Bv'\in\partial^+\BC^{(n)}_{\ell}(\Bw)}|\BG^{(n)}_{\BC^{(n)}_L(\Bu)}(\Bv',\By,\zeta)|\label{eq:GRI}
\end{align}

We also have the Geometric resolvent inequality for the eigenfunctions. Consider an eigenfunction $\BPsi_i$ of the operator $\BH^{(n)}_{\BC^{(n)}_L(\Bu)}$ then for any given complex $\zeta$ which  is not eigenvalue of the operator $\BH^{(n)}_{\BC^{(n)}_L(\Bu)}$ we have that
\begin{equation}\label{eq:GRE.eigenfunction}
\BPsi_i(\Bx)=\sum_{\Bv,\Bv'\in\partial\BC^{(n)}_{\ell}(\Bw)}\BG^{(n)}_{\BC^{(n)}_{\ell}(\Bw)}(\Bx,\Bv,\zeta)\BPsi_i(\Bv'), \quad \Bx\in\BC^{(n)}_{\ell}(\Bw),
\end{equation}
yielding the geometric resolvent inequality for eigenfunctions
\begin{equation}\label{eq:GRI.eigenfunction}
|\BPsi_i(\Bx)|\leq  |\partial\BC^{(n)}_{\ell}(\Bw)| \left(\max_{\Bv\in\partial^-\BC^{(n)}_{\ell}(\Bw)}\BG^{(n)}|\BG^{(n)}_{\BC^{(n)}_{\ell}(\Bw)}(\Bv,\Bx,\zeta)|\right)\max_{\Bv'\in\partial^+\BC^{(n)}_{\ell}(\Bw)}|\BPsi_i(\Bv')|
\end{equation}
The proof of these bounds on the geometric resolvent estimates can be found in the book by Kirsch \cite{Kir08}.

\subsection{Radial descent bounds}
Here we present  a radial  descent bound which encapsulates previous  method of the multi-particle multi-scale analysis such as method developed for example in the papers \cites{DK89,CS09b}.

The new method originally extended in the work \cites{Chu10a, Chu10b} do not need a difficult property  of separable pairs of cubes and use a much weaker assumptions on the conditional distribution  of the sample mean. The radial descent bound uses a sub-harmonicity property. For that reason, we introduce

\begin{definition}
Let $f:\Lambda\rightarrow \DC$ be a bounded function on a subset $\Lambda\subset\DZ^D$, $D\geq 1$. The function $f$ is called $(\ell,q,S,c)$-sub-harmonic, with $\ell_in\DN$ $q\in(0,\infty)$, $S\subset\Lambda$, $c\geq 1$, if for any $\Bx\in\Lambda\setminus S$, such that, $C_{\ell}(\Bx)\subset\Lambda$ we have that
\[
|f(\Bx)|\leq q \max_{|\By-\Bx|=\ell}|f(\By)|,
\]
while for $s\in S$
\[
|f(\Bx)|\leq q \max_{\By\in\Lambda, \ell\leq |\By-\Bx|\leq (1+c)\ell}|f(\By)|
\]
\end{definition}
Now, we can state the crucial Lemma of the Section,

\begin{lemma}\label{lem:subharmonicity}
Let $f$ be an $(\ell,q,S,c)$-subharmonic  function on $\BC^{(n)}_L(\Bu)\subset\DZ^D$. Assume that the $c\ell$-neighborhood of the set $S$ can be covered by a collection  $\CA$ of annuli
\[
A_i=\BC^{(n)}_{b_i}(\Bu)\setminus \BC^{(n)}_{a_i}(\Bu),
\]
of total width$W_{\CA}$. Then for any 
\[
r\in[W_{\CA}+\ell, L-W_{\CA}+\ell],
\]
\[
\max_{\Bx\in \BC^{(n)}_r(\Bu)}|f(\Bx)|\leq q^{[(L-r-W(A))/{\ell}]-1}\max_{\By\BC^{(n)}_L(\Bu)}|f(\By)|
\]
In particular
\[
|f(0)|\leq q^{[(L-W(A))/\ell]-1}\max_{\By\in\BC^{(n)}_L(\Bu)}|f(\By)|
\]
\end{lemma}
\begin{proof}
We will just consider non-negative function $f$ instead we can use $|f|$.
Let a function $f: \BC^{(n)}_L(0)\rightarrow \DR_+$ be $(\ell,q,S,c)$-subharmonic and introduce spheres $
S_r=\{\Bx: |\Bx|=r\}$ and set 
\[
S'=\{\Bx: S_{|\Bx|}\cap S \neq \emptyset\}\quad S''=\bigcup_{r:S_r\cap S\neq \emptyset}\bigcup_{j=0}^{c\ell} S_{r+j}
\]
Also define 
\[
\CR=\{r\geq 0: \BC^{(n)}_r(0)\subset \BC^{(n)}_L(0)\setminus  S''\}.
\]
Recall that if $|\Bx|=r\in\CR$, then 
\begin{equation}\label{eq:subharmonicity}
f(\Bx)\leq q\max_{\By: |\By-\Bx|\leq \ell} f(\By)\leq q \max_{\By: |\By|\leq r+\ell}f(\By).
\end{equation}
Further for any $\Bu$, with $|\Bu|\in[r-c\ell,r]$ the Definition of the sets $\CR$ and $S''$ implies  that $\Bu\notin S$, so that
\begin{equation}\label{eq:subharmonicitytwo}
f(\Bu)\leq \max_{\By:|\By-\Bu|\leq \ell} f(\By)\leq q\max_{\By:|\By|\leq |\Bu|+\ell}f(\By)
\end{equation}
\[
\leq q:\max_{\By:|\By|\leq r+\ell} f(\By).
\]
The combination of the equations \eqref{eq:subharmonicity} and \eqref{eq:subharmonicitytwo} gives for any $r\in\CR$ 
\begin{equation}\label{eq:subharmonicitythree}
\max_{\Bu\in\BC^{(n)}_{r}(0)} f(\Bu)\leq q \max_{\By\in \BC_{r+\ell}(0)}f(\By)
\end{equation}
We define a sequence of points  
\[
\{r_n\geq 0,: 0\leq n\leq M\}
\]
by recursion 
\[
r_0=L, \quad r_n=max\{r\in\CR: r\leq r_{n-1}-\ell\},\quad 1\leq n\leq M,
\]
with some $M=M(R)$ and assume $r_{M+1}=0$. One can see that when $r_n\geq 0$, either $r_n=r_{n-1}-\ell$ or $J_n:=[r_n,r_{n-1}-\ell]\subset \CR^c$. The total length of all non-empty intervals $J_n$ is bounded by the total width $W(S'')$ of annuli covering $S''$. Now we have that 
\begin{align*}
L&=\sum_{n=0}^M(r_{n-1}-r_n)\\
&=(r_M-r_{M+1}) + \sum_{n:J_n\neq \emptyset}(r_{n-1}-r_n)+\ell\card\{n:r_n=r_{n-1}-\ell\}\\
&\leq \ell+W(S'')+\ell\card\{n: r_n=r_{n-1}-\ell\}
\end{align*}
getting 
\begin{align*}
M&\geq  \card\{n: r_n=r_{n-1}-\ell\}\geq \frac{L-W(S'')}{\ell}-1\\
&\geq \frac{L-W_{\CA}}{\ell}-1
\end{align*} 
Next $W(S'')\leq W_{\CA}$, since  the annuli, $A_i\in\CA$ by assumption, covers the $c\ell$-neighborhood of the set $S$. Thus, 
\[
f(0)\leq q^{M}\|f\|_{\infty}.
\]
If the sequence $\{r_n\}$ is stopped at $n=M'$ with $r_{M'+1}\leq r*$ for some $r\geq W_{\CA}+\ell\geq W(S'')+\ell$, then we have  
\[
f(r)\leq q^{M'}\|f\|_{\infty},\quad  M'\geq \frac{L-r-W_{\CA}}{\ell}-1
\]
which ends the proof of Lemma \ref{lem:subharmonicity}.
\end{proof} 
We use Lemma \ref{lem:subharmonicity} to prove the following result. 

\begin{lemma}\label{lem:annuli}
Let consider a lattice Schr\"odinger operator $\BDelta +\BV(\Bx)$ restricted on a cube $\BC^{(n)}_{L_{k+1}}(\Bu)\subset\DZ^D$; $D\geq 1$. For any fixed energy $E\in I$ and suppose that the $cL_k$-neighborhood of all $(E,m)$-singular  cubes of radius $L_k$ inside $\BC^{(n)}_{L_{k+1}}(\Bu)$ can be covered by a collection $\CA$ of annuli, $A_i=\BC^{(n)}_{b_i}(\Bu)\setminus \BC^{(n)}_{a_i}(\Bu)$ of total width $W_{\CA}$. Further suppose that $\BC^{(n)}_{L_{k+1}}(\Bu)$ does not contains any $E$-resonant cube of radius $L\geq L_k$. Then for any fixed  value of the  constant $\tilde{c}$ and $L_0\geq L_0^*(c,\Bu)$ large enough the cube  $\BC^{(n)}_{L_{k+1}}(\Bu)$ is $(E,m)$-NS.
\end{lemma}
\begin{proof}
Later
\end{proof}

\subsection{Localization bounds for decomposable systems PI cubes}
At the beginning of the multi-scale analysis we need the systems of the particles to be decomposable into two non-interacting subsystems. Such systems of cubes will be refer as PI cubes.

\begin{definition}
\begin{enumerate}
\item[(i)]
Let $n'\in\{1,\ldots,N-1\}$, $\geq 0$ and  $\Bu'=(u_1,\ldots,u_{n'})\in\DZ^{n'd}$. Given a positive $m$ the $n'$-particle cube $\BC^{(n')}_{L}(\Bu')$ is  called $m$-tunnelling ($m$-T) if $\exists E\in I$ and two $2NL_k$-distant $n'$-particle cubes  $\BC^{(n')}_{L_k}(\Bv_1)$, $\BC^{(n')}_{L_{k-1}}(\Bv_2)\subset\BC^{(n')}_{L_k}(\Bu')$ which are $(E,m)$-S.
\item[(ii)] An $n$-particle cube $BC^{(n)}_{L_k}(\Bu)$ is called $m$-partially tunnelling ($m$-PT) if for some permutation  $\tau\in\mathfrak{S}_N$ acting on $\Bu=(u_1,\ldots,u_n)$ and some $n',n''\geq 1$, it admits a representation 
\[
\BC^{(n)}_{L_k}(\tau(\Bu))=C^{(n')}_{L_{k}}(\Bu')\times \BC^{(n'')}_{L_k}(\Bu'')
\]
$\Bu'=(u_1,\ldots,u_{n'})$, $\Bu''=(u_{n'+1},\ldots,u_n)$ where either $\BC^{(n')}_{L_k}(\Bu')$ or $\BC^{(n'')}_{L_k}(\Bu'')$ is $m$-T. Otherwise it is called $m$-NPT.
\end{enumerate}
\end{definition} 

\begin{lemma}\label{lem:tunnelling}[Cf. \cite{Eka20}]
Let $E\in I$ and consider an $n$-particle cube $\BC^{(n)}_L(\Bu)$,  $\BC^{(n)}_L(\Bu)=\BC^{(n')}_L(\Bu')\times \BC^{(n'')}_L(\Bu'')$, with $n',n''\geq 1$ and a sample  of potential $V(\cdot,\omega)$ such that $\dsknprimeN$  holds true for $1\leq n'\leq n-1$ and 
\begin{enumerate}
\item[(i)]
$\rho(\varPi\BC^{(n')}_{L_k}(\Bu'),\varPi\BC^{(n'')}_{L_k}(\Bu''))\geq 2L_k+r_0$.\\
\item[(ii)]
$\BC^{(n)}_{L_k}(\Bu)$ is $E$-NR
\item[(iii)]
$\BC^{(n')}_{L_k}(\Bu')$ and $\BC^{(n'')}_{L_k}(\Bu'')$ are $m$-NT
\end{enumerate}
Then  $\BC^{(n)}_{L_k}(\Bu)$, is  $(E,m)$-NS
\end{lemma}

\begin{lemma}\label{lem:M}[Cf. \cites{Eka20, Chu10b}]
\[
\prob\{\text{$\BC^{(n)}_{L_k}(\Bu)$ is $m$-PT}\}\leq \frac{1}{2} L_k^{-p2^{N-(n-1)+1}+ 2Nd\alpha}
\]
\end{lemma}
Introduce  the following variable, depending on the sample $V(\cdot,\omega)$

$M^{PI}(\BC^{(n)}_{L_{k+1}}(\Bu),I)$=The maximal number of $2NL_k$-distant PI cubes of radius $L_k$ contains in $\BC^{(n)}_{L_{k+1}}(\Bu)$ which are $(E,m)$-S with  $E\in I$.\\

$M^{FI}(\BC^{(n)}_{L_{k+1}}(\Bu),I)$= The maximal number of  $2NL_k$-distant FI cubes of radius $L_k$ contains in $\BC^{(n)}_{L_{k+1}} (\Bu)$ are $(E,m)$-S with $E\in I$ 

\begin{lemma}\label{lem:prob.MPI}
Let $p\in(Nd \alpha/(2-\alpha),\infty)$ and $L_0^{p(2-\alpha)-Nd\alpha}\geq 2$ then
\[
\prob\left\{ M^{PI}(\BC^{(n)}_{L{k+1}}(\Bu),I)\geq 2\right\}\leq \frac{1}{2} L_k^{-4p+2Nd\alpha}\leq \frac{1}{8} L_{k+1}^{-2p}
\]
\end{lemma} 

\subsection{Localization bounds for non decomposable systems of FI cubes}

\begin{lemma}\label{lem:prob.MFI}[Cf. \cites{Eka20,Chu10b}]
If $p\in(2Nd\alpha/(2-\alpha),\infty)$, then
\begin{align*}
\prob\left\{ M^{FI}(\BC^{(n)}_{L_{k+1}}(\Bu),I)\geq 4\right\}&\leq \frac{1}{4!} L_{k+1}^{-4p/\alpha +4Nd}\\
&\leq \frac{1}{4!} L_{k+1}^{-2p}
\end{align*}
\end{lemma}
The proof uses the fact that distant projections of FI cubes lead to independent samples of the potential $V$. We see that with high probability $M^{PI}+M^{FI}\leq 4$

\begin{theorem}\label{thm:FI.cubes}
Suppose that the bound $\dsknN$ holds true for some $k\geq 0$. Then the bound $\dskonN$ also holds true for all $2NL_{k+1}$-distant pairs of FI cubes.
\end{theorem}  

\begin{proof}
Let consider two $2NL_{k+1}$-distant FI cubes $\BC^{(n)}_{L_{k+1}}(\Bx)$, $\BC^{(n)}_{L_{k+1}}(\By)$. Both of them are $(E,m)$-singular for some $E\in I$ only if at least one of the following events occurs:
\begin{enumerate}
\item[(1)]
for some  $E\in I$, both $\BC^{(n)}_{L_{k+1}}(\Bx)$ and $\BC^{(n)}_{L_{k+1}}(\By)$ are $E$-R\\
\item[(2)]
$M^{FI}(\BC^{(n)}_{L_{k+1}}(\Bx),I)\geq 4$ or $M^{PI}(\BC^{(n)}_{L_{k+1}}(\Bx),I)\geq 2$,\\
\item[(3)]
$M^{FI}(\BC^{(n)}_{L_{k+1}}(\By),I)\geq 4$ or $M^{PI}(\BC^{(n)}_{L_{k+1}}(\By),I)\geq 2$
\end{enumerate}
By Theorem \ref{thm:Wegner} the probability of the envents (1) can be bounded by $\frac{1}{4} L_{k+1}^{-2p}$. Further, it follows from  Lemma \ref{lem:prob.MFI} and lemma \ref{lem:prob.MPI} that the probability of the event (2) as well as the event (3) is bounded by
\[
\frac{1}{8} L_{k+1}^{-2p}+ \frac{1}{4!}L_{k+1}^{-2p}\leq \frac{1}{4}L_{k+1}^{-2p}
\]
Thus the cubes $\BC^{(n)}_{L_{k+1}}(\Bx)$ and $\BC^{(n)}_{L_{k+1}}(\By)$ are $(E;m)$-S for some $E\in I$ with probability less than $(\frac{1}{4}+\frac{2}{4}) L_{k+1}^{-2p}\leq L_{k+1}^{-2p}$ as required.
\end{proof}

\subsection{Pairs of decomposable PI cubes and mixed pairs}
  
\begin{lemma}\label{lem:mixed.pairs}
Assume that the bound $\dsknN$ holds true foe some $k\geq 0$. Then the bound $\dskonN$ also holds true  for all  $2NL_{k+1}$-distant pairs of cubes $\BC^{(n)}_{L_{k+1}}(\Bx)$ and $\BC^{(n)}_{L_{k+1}}(\By)$ where one at least is partially interactive.
\end{lemma}
\begin{proof} We can assume that the cube  $\BC^{(n)}_{L_{k+1}}(\By)$ is PI. Since the cubes $\BC^{(n)}_{L_{k+1}}(\Bx)$ and $\BC^{(n)}_{L_{k+1}}(\By)$ are $(E,m)$-singular simultaneously only if at least one of the following  events occurs:
\begin{enumerate}
\item[(1)]
for some  $E\in I$ both  $\BC^{(n)}_{L_{k+1}}(\Bx)$ and $\BC^{(n)}_{L_{k+1}}(\By)$ are $E$-R\\
\item[(2)]
$\BC^{(n)}_{L_{k+1}}(\Bx$ is FI and $M^{FI}(\BC^{(n)}_{L_{k+1}}(\Bx),I)\geq 4$ or $M^{PI}(\BC^{(n)}_{L_{k+1}}(\Bx),I)\geq 2$\\
\item[(3)]
The PI cube $\BC^{(n)}_{L_{k+1}}(\By)$ is partially tunnelling 
\end{enumerate}
The probabilities of the events (1) and (2) can be estimated exactly as in the proof of Lemma \ref{lem:prob.MFI} so that their sum is bounded by $\frac{1}{2}L_{k+1}^{-2p}$. By inductive assumption  $\dskprimenN$, $k'\geq 0$ on systems with $n\leq N$ particles, the event (3) has probability bounded by $\frac{1}{4}L_{k+1}^{-2p}$. Finally the Lemma follows by combining the above estimates.
\end{proof}  

\subsection{Conclusion }  
Finally, we have proved that  for any $n=1,\ldots,N$ the multi-scale analysis estimates $\dsknN$ holds ture for all $k\geq0$.

\section{Proofs of the multi-particle localization results}\label{sec:proofs.results}
The derivation  of the spectral exponential localization given the results of the multi-scale analysis has been obtained originally  in the paper \cites{DK89,FMSS85} for the single-particle models and in \cites{CS09a,CS09b} in the multi-scale Anderson model. In this paper we prove using the bounds  on the resonances for large multi-particle systems the localization for distant pairs of cubes without the condition of separability. This property is encountered  by a  new multi-scale analysis scheme using distant pairs of cubes.

\subsection{Proof of the exponential localization}
It is plain that for almost every  generalized eigenfunction 
$\BPsi$ of the operator $\BH^{(n)}(\omega)$  is polynomially bounded 
\[
|\BPsi(\Bx)|\leq Const |\Bx|^{A}, \quad A(0,\infty)
\]
So it suffices to show that every polynomially bounded  solution of the equation $\BH^{(N)}\BPsi=E\BPsi$ with $E\in I$ is exponentially decaying fast at infinity. For any non-zero eigenfunction $\BPsi$ there exists a point $\Bu\DZ^{Nd}$ with $\BPsi(\Bu)\neq 0$. Let $k_0\geq 0$ be an integer such that the cube $\BC^{(N)}_{L_{k_0}}(0)$ contains  all points $\tau(\Bu)$, $\tau\in S_n$.
First observe that for some $k_1\geq k_0$  and $k\geq k_1$ the cubes $\BC^{(N)}_{L_k}(0)$   must be $(E,m)$-S. Indeed, assume the contrary, then there exists an infinite number of values of $k$ (thus an arbitrarily large values of  $L_k$) such that $\BC^{(N)}_{L_k}(0)$ is $(E,m)$-NS with $\Bu\in\BC^{(N)}_{L_{k-1}}(0)$. Thus the non-singularity property implies that for all points $\Bx\in\BC^{(N)}_{L_{k-1}}(0)$, including $\Bx=\Bu$, we get 
\begin{align*}
|\BPsi(\Bx)|&\leq \ee^{-\gamma(m,L_k,N)L}|\partial^-\BC^{(N)}_{L_k}(0)|\max_{\By\in\partial^+\BC^{(N)}_{L_k}(0)}|\BPsi(\By)|\\
&\leq \ee^{-m L_k} O\left( L_k^{(N-1)d+A}\right)\rightarrow 0 \quad L_k\rightarrow \infty
\end{align*} 
leading to $\Psi(u)=0$ in contradiction  with the choice of the point $\Bu$. Now set $\alpha'=9/8\leq \alpha$ for all $k\geq k_1$ we consider  the events:
\[
\CE_k=\left\{\exists \lambda\in I: \BC^{(N)}_{2L_{k+1}^{\alpha'}}(0) \text{contains two $2NL_k$-distant $(\lambda,m)$-S cubes of radius $L_k$}\right\}
\]
and   
\[
\tilde{\Omega}:= \bigcup_{k\geq k_1}\bigcap_{j\geq k}\left(\Omega\setminus\CE_k\right)
\]
Since $\prob\{\CE_k\}=O(L_k^{2Nd\alpha\alpha'-2p})$ and $p\geq3Nd\alpha\geq Nd\alpha\alpha'$. It follows from the Borel-Cantelli Lemma that $\prob\{\tilde{\Omega}\}=1$ and for every $\omega\in\tilde{\Omega}$ there exists $k_2\geq k_1$ such that no cubes $\BC^{(N)}_{L_{j+1}}(0)$ with $j\geq k_2$ contains a pair of $(E,m)$-S cubes of radius $L_j$ at distance $\leq 2NL_j$. Next, we are able to introduce annuli 
\[
A_k=\BC^{(N)}_{L_{k+1}^{\alpha'}}(0)\setminus  \BC^{(N)}_{L_k^{\alpha'}}(0),\quad k\geq 0,
\] 
and let $\Bx\in\DZ^{Nd}\setminus \BC^{(N)}_{L_{k_2}^{\alpha'}}(0)=\bigcup_{k\geq k_2} A_k$.

Assume that $\Bx\in A_j$ for some $j$, so that the cubes $\BC^{(N)}_{L_{k+1}}(0)$ and $\BC^{(N)}_{L_j}(0)$ are $2NL_j$-distant and $\omega\in\tilde{\Omega}$, such that one of them is $(E,m)$-NS since $k\geq k_2$, $\BC^{(N)}_{L_j}(0)$ is $(E,m)$-S, Hence $\BC^{(N)}_{L_j}(\Bx)$ is $(E,m)$-NS. Moreover, all the cubes of radius $L_k$ outside $\BC^{(N)}_{L_k+2NL_k}(0)$ are also $(E,m)$-NS. This implies that the function $\BPsi$ is $(L_k,q)$-subharmonic in the cube $\BC^{(N)}_R(\Bx)$ with 
\[
q=\ee^{-m(L_k+O(\ln L_k))}\leq\ee^{-m(L_k+\frac{1}{2}L_k^{3/4})},
\]
if $L_k$ is large enough and 
\[
R=|\Bx|-(L_k+2NL_k),  \quad |\Bx|\geq L_k^{9/8}\geq L_k
\]

It is easy to see that $R/|\Bx|=1-O(L_k^{-1/8})$ and applying Lemma \ref{lem:annuli}, we obtain for $L_k$ large enough
\[
-\frac{\ln|\BPsi(\Bx)|}{|\Bx|}\geq m\frac{L_k+L_k^{3/4}}{L_k}\cdot\left(1-O(L_k^{-1/8})\right)\geq m
\]
which ends the proof of Theorem \ref{thm:main.result.exp.loc}.

\subsection{Proof of the dynamical localization}
We are now ready to prove the dynamical localization with method developed for example in the particle localization theory with the multi-scale analysis, see for example \cites{Kir08, Sto01,Chu10a}. Before we need one more assumption on the random Hamiltonian,
\begin{equation}\label{eq:trace.prob}
\exists \kappa=\kappa(I,N,d)\in(0,\infty): \prob\{tr(P_I(\BH^{(N)}_{\BC^{(N)}_L(\Bu)}(\omega)))\geq CL^{\kappa Nd}\}\leq L^{-B'}
\end{equation}
Observe that 
\begin{align*}
tr(\BH^{(N)}_{\BC^{(N)}(\omega)})&=tr\left(\BDelta_{\BC^{(N)}_L(\Bu)}+\BU_{\BC^{(N)}_L(\Bu)}\right)
\\ &+tr\left(\BV_{\BC^{(N)}_L(\Bu)}(\omega)\right)\\
\end{align*}  
\[
tr(\BV_{\BC^{(N)}_L(\Bu)}(\omega))=\sum_{\Bx\in \BC^{(N)}_L(\Bu)} \sum_{j=1}^NV(x_j,\omega)
\]

\subsubsection{Bad and good events}
Fix $s\in(0,\infty)$ and assume that
\[
p\geq \max\left\{2Nd\alpha/(2-\alpha),(3Nd\alpha+\alpha s)/2\right\}
\]
and set $b=b(p,N,d,\alpha):= 2p-2Nd\alpha$ for each $j\geq 1$ consider the events 
\[
 S_j=\{\text{$\omega: \exists E\in I, \exists \By,\Bz\in\BC^{(N)}_{4(N+1)L_{j+1}}(0)| d_S(\By,\Bz)\geq 2NL_j $
 and $\BC^{(N)}_{L_j}(\By), \BC^{(N)}_{L_j}(\Bz)$ are $(E,m)$-S}\}
\]
and in the case where the random potential $V$ is not bounded  from below
\[
\tau_j:=\left\{\omega: tr(P_I(\BH^{(N)}_{\BC^{(N)}_{L_{j+2}}(\Bu)}))\geq CL_{j+2^{\kappa Nd}}\right\}
\]
Further, for $k\geq 1$ denote
\[
\Omega_k^{(bad)}=\bigcup_{j\geq k}(S_j\cup\tau_j),
\]
and consider the  annuli 
\[
M_k=\BC^{(N)}_{4(N+1)L_{k+1}}(0)\setminus \BC^{(N)}_{4(N+1)L_k}(0),
\]

\begin{lemma}\label{lem:bad.event}
under assumptions  $\cmu$ with $B'\geq 2p-2Nd\alpha$,
\[
\forall k\geq 1,\quad \prob\{\Omega_k^{(bad)}\}\leq c(\alpha,d,p,N)L_k^{-(2p-Nd\alpha)}
\]
\end{lemma} 
\begin{proof}
The number of pairs $\By,\Bz\in\BC^{(N)}_{4(N+1)L_{j+1}}(0) $ figuring in the definition of the event $S_j$ is bounded by $\frac{1}{2}(4(N+1)L_{j+1})^{2Nd}=C(N,d)L_j^{2Nd \alpha}$ and from $\dsknN$ we have that
\[
\prob\{\text{$\BC^{(N)}_{L_j}(\By), \BC^{(N)}_{L_j}(\Bz)$ are $(E,m)$-singular}\}\leq L_j^{-2p},
\]
while for the event $\tau_j$ we have the bound $\cmu$. Now we require the exponent $B'$ to be large enough, so that $\prob\{\tau_j\}\leq L_j^{-2p+2Nd\alpha}$ and  

\begin{align*}
\prob\left\{\Omega_k^{(bad)}\right\}&\sum_{j\geq k}(\prob\{S_j\}+\prob\{\tau_j\})\leq \sum_{j\geq k} Const L_j^{2Nd\alpha}\cdot 2L_j^{-2p}\\
&= L_k^{-b}\left[1+\sum_{j\geq k} L_k^b L_k^{-b\alpha^{j-k}}\right]\leq Const L_k^{-(2p-2Nd\alpha)}\\
\end{align*}
\end{proof}  

\subsubsection{Centers of localization}
By Theorem \ref{thm:main.result.exp.loc} there exists $\Omega_1\subset\Omega$ with $\prob\{\Omega_1\}=1$ such that for any $\omega\in\Omega_1$ the  spectrum of $\BH^{(N)}(\omega)$ in $I$ is pure point. Now fix $\omega\in\Omega_1$ and let $\Phi_n(\omega)$ be a normalized eigenfunction of $\BH^{(N)}(\omega)$, with eigenvalue $E_n(\omega)$ we call a center of localization  for $\Phi_n$ every point $\Bx_n(\omega)\in\DZ^{Nd}$ such that 
\begin{equation}\label{eq:CL}Si
|\Phi_n(\Bx_n(\omega))|=\max_{\By\in\DZ^{Nd}}|\Phi_n(\omega)(\By)|.
\end{equation}
Since $\Phi_n\in\ell^2(\DZ^{Nd})$, such centers exists and since $\|\Phi_n\|=1$, the number of centers of localization $\Bx_{n,a}$ for a given $n$ must be finite. 

\begin{lemma}\label{lem:CL}
There exists $k_0$ such that for all $\omega\in\Omega_1$ and $k\geq k_0$, if one of the centers of localization $\Bx_{n,a}$ for an eigenfunction $\Psi_n$ with eigenvalue $E_n\in I$ belongs to a cube $\BC^{(N)}_{L_k}(\Bx)$ then the cube $\BC^{(N)}_{L_{k+1}}(\Bx)$ is $(E_n,m)$-S.
\end{lemma} 
\begin{proof}
Assume the contrary. Then the GRI for eigenfunction  combine with the non-singularity condition implies that, for $L_k$ large enough,
\begin{gather*}
|\BPsi_n(\Bx_{n,a})|\leq \ee^{-\gamma(m,L_{k+1},N)L_{k+1}}|\partial\BC^{(N)}_{L_k+1}(\Bx)|\\
\times \max_{\By\in\partial^+\BC^{(N)}_{L_{k+1}}(\Bx)}|\BPsi_n(\By)|\\
\leq \max_{\By\in \partial^+\BC^{(N)}_{L_{k+1}}(\Bx)}|\BPsi_n(\By)|,
\end{gather*}

which contradicts the definition of  a center of localization.
\end{proof}
 Denote by $\Bx_{n,1}$, the center of localization closed to the origin. This center might be not unique and it does not matter. Fix $k_0$ as in  Lemma \ref{lem:CL} and define 
\[
\Omega_k^{(good)}=\Omega_1\setminus \Omega_k^{(bad)}.
\]

 \begin{lemma}\label{lem:CL.bound}
There exists $j_0=j_0(m,\alpha,d)$ large enough such that for $j\geq j_0$, $j\geq k$ and $\Bx_{n,1}\in\BC^{(N)}_{j+1}(0)$
\[
\|\left(1-\Bone_{\BC^{(N)}_{4(N+1)L_{j+2}(0)}}\Phi_n\right)\|\leq \frac{1}{4}
\]
\end{lemma}

\begin{proof}
Using the annuli $M_k$, we can write 
\begin{gather*}
\left\|(1-\Bone_{\BC^{(N)}_{4(N+1)L_{j+1}(0)} })\Phi_n\right\|^2=\sum_{i\geq j+2}\|\Bone_{M_i}\Phi_n\|^2
\\
=\sum_{i\geq j+2}\sum_{\By\in M_i}|\Phi_n(\By)|^2\\
\end{gather*} 
Fix $i\geq j+2$. The cube  $\BC^{(N)}_{L_{i-1}}(0)\subset \BC^{(N)}_{L_{j+1}}(0)$, contains the center of loaclization $\Bx_{n,1}$, so by Lemma \ref{lem:CL}, $\BC^{(N)}_{L_j}(0)$ is $(E_n,m)$-singular. By the construction of the event $\Omega_k^{(good)}$, $k\leq j$  the cube $\BC^{(N)}_{L_i}(\By)$ must be $(E_n,m)$-NS. So, this implies by the GRI for eigenfunctions that $|\Phi_n(\By)|^2\leq \ee^{-2m L_i}$. Finally the claim follows from a polynomial bound on the number of terms in the sum
\[
\sum_{M_i}|\Phi_n(\By)^2
\]
\end{proof} 
\begin{lemma}\label{lem:card.CL}
There exists $C_4=C_4(m,d,k)$ susch that for $\omega\Omega_k^{(good)}$ and $j\geq k$, the following bound holds true:
\[
\card\left\{ n: \Bx_n\in\BC^{(N)}_{L_{j+1}}(0)\right\}\leq C_4 L_{j+1}^{\alpha k,d}.
\]
\end{lemma} 
\begin{proof}
We have 
\begin{gather*}
\sum_{\Bx_{n,1}\in\BC^{(N)}_{L_{j+1}}(0)}\left(\Bone_{\BC^{(N)}_{L_{j+2}}(0)}P_I(\BH^(N))\Bone_{\BC^{(N)}_{L_{j+2}}(0)}\Phi_n,\Phi_n\right)\\
\leq tr\left(\Bone_{\BC^{(N)}_{L_{j+2}}(0)}P_I(\BH^{(N)})\right)
\end{gather*} 
we will show that each term in the left hand side is bigger than $1/2$/ Indeed, by Lemma \ref{lem:CL.bound}
\begin{gather*}
\left(\Bone_{\BC^{(N)}_{j+2}(0)}P_I(\BH^{(N)}) \Bone_{\BC^{(N)}_{j+2}(0)}\Phi_n,\Phi_n\right)\\
=\left(\Bone_{\BC^{(N)}_{j+2}(0)}P_I(\BH^{(N)}) \Bone_{\BC^{(N)}_{j+2}(0)}\Phi_n,\Phi_n\right)-\left(\Bone_{\BC^{(N)}_{L_{j+2}}(0)} P_I(\BH^{(N)})\right)\\
\left(1-\Bone_{\BC^{(N)}_{j+2}(0)}\Phi_n,\Phi_n\right)\\
\geq \left(\Bone_{\BC^{(N)}L_{j+2}(0)}\Phi_n,\Phi_n\right)-\frac{1}{4}=(\Phi_n,\Phi_n)-\left((1-\Bone_{\BC^{(N)}_{L_{j+2}}(0)}\right)\Phi_n,\Phi_n)-\frac{1}{4}\geq 1/2.
\end{gather*}
\end{proof}

\subsubsection{Eigenfunction correlator bounds}
\begin{lemma}\label{lem:EFC}

There exists an integer $k_1=k_1(\kappa,L_0)$ such that for  all $k\geq k_1$, $\omega\in\Omega_k^{(good)}$ and $\Bx\in M_k$
\[
|f(\BH^{(N)}(\omega))(\Bx,0)|\leq \sum_{n:E_n\in I}\|f\|_{\infty}
\]
\end{lemma} 
\begin{proof} 
Without loss of generality, one can assume that $\|f\|_{\infty}\neq 0$,
\begin{gather}
\|f\|_{\infty} |f(\BH^{(N)}(\omega))(\Bx,0)|\leq \sum_{n:E_n\in I}\|f\|_{\infty}|f(E_n)|\notag\\
|\Phi_n(\Bx)| |\Phi_n(0)|\leq Const L_{k+1}^{\alpha\kappa Nd}\ee^{-mL_k}\notag\\
\leq \sum_{\substack{n: E_n\in I\\ \Bx_{n,1}\in M_k(0)}} |\Phi_n(\Bx)||\Phi_n(0)|\leq Const L_{k+1}^{\alpha \kappa Nd}\ee^{-mL_k}\notag\\
\leq \frac{1}{2}\ee^{-mL_k/2} \label{eq:EFC1}
\end{gather}   
Since one of the cubes $\BC^{(N)}_{L_k}(\Bx)$, $\BC^{(N)}_{L_k}(0)$ must be  $(E_n,m)$-NS, indeed the cubes are $2N L_k$-distant. Next, fix any $j\geq k+1$ and consider the sum with $\Bx_{n,1}\in M_j(0)$. The cubes $\BC^{(N)}_{L_j}(0)$ and $\BC^{(N)}_{L_j}(\Bx_{n,1}$ are $2NL_j$-distant and Lemma \ref{lem:CL.bound}, for $k$ (hence $j$) large enough the cube $\BC^{(N)}_{L_j}(\Bx_{n,1})$ is $(E_n,m)$-S, so that for $\omega\in\Omega_k^{(good)}$ the cube $\BC^{(N)}_{L_j}(0)$ must be $(E_n,m)$-NS. Therefore 
\[
|\Phi_n(0)|\leq Const\ee^{-mL_j}
\]
By Lemma \ref{lem:EFC}, we have that for $L_k$ large enough,  
\begin{align}
\sum_{j\geq k}\sum_{\substack{n: E_n\in I\\ \Bx_{n,1}\in M_j(0)}}|\Phi_n(\Bx)| |\Phi_n(0)|&\leq Const\sum_{j\geq k}\ee^{-m L_j} L_j^{\alpha\kappa Nd}\\
&\leq \frac{1}{2}\ee^{-mL_k/2} \label{eq:EFC2}
\end{align}
The Lemma then follows from \eqref{eq:EFC1} and \eqref{eq:EFC2}.
\end{proof} 

\begin{lemma}\label{lem:dynamical.bound}
Let $k_1$ be as in Lemma \ref{lem:EFC}. Then for any $k\geq k_1$ and $\Bx\in M_k$,
\begin{align}
&\esm\left[\|\Bone_{\BC^{(N)}_{L_k}(\Bx)}f(\BH^{(N)}(\omega))\Bone_{\BC^{(N)}_{L_k}(0)}\|\right]\leq \|f\|_{\infty}\notag
&\left( C_{L_k}^{-2p+2Nd\alpha} + \ee^{-m L_k/2}\right)\label{eq:dynamical.bound}
\end{align}
\end{lemma}

\begin{proof}
Using Lemma \ref{lem:EFC} and Lemma \ref{lem:bad.event}, we can  write
\begin{align*}
& \esm\left[\|\Bone_{\BC^{(N)}_{L_k}(\Bx)}f(\BH^{(N)}(\omega))\Bone_{\BC^{(N)}_{L_k}(0)}\|\right]\\
&=\esm\left[\Bone_{k}^{(bad)}\|\Bone_{\BC^{(N)}_{L_k}(\Bx)}f(\BH^{(N)}(\omega))\Bone_{\BC^{(N)}_{L_k}(0)}\|\right]\\
&+\esm\left[\Bone_{\Omega_k}^{(good)}\| \Bone_{\BC^{(N)}_{L_k}(\Bx)}f(\BH^{(N)}(\omega))\Bone_{\BC^{(N)}_{L_k}(0)}\|\right]\\
&\leq \|f\|_{\infty}\left(\prob\left\{\Omega_k^{(bad)}\right\}+ \ee^{-m L_k/2}\right)\leq \|f\|_{\infty}\left( CL_k^{-2p+2Nd\alpha}+\ee^{-mL_k/2}\right)\\
\end{align*}
\end{proof}  

\subsubsection{Conclusion}
Fix a set $\BK\in\DZ^{Nd}$ and find $k\geq k_1$ such that$\BK\in\subset  \BC^{(N)}_{L_k}(0)$. Then 
\begin{align*}
&\esm\left[\| \BX^s f(\BH^{(N)}(\omega))\Bone_{\BK}\|\right]\leq C_{Nd} L_{k}^s\\
& +\sum_{j\geq k}\esm\left[\|\Bx^s\Bone_{M_j} f(\BH^{(N)}(\omega))\Bone_{\BK}\|\right]\\
&\leq c(k)+\sum_{j\geq k} C_{Nd}L_{j+1}^s\left(\sum_{\Bw\in M_j}\esm\left[\|\Bone_{\BC^{(N)}_{L_k}(\Bw)}f(\BH^{(N)}(\omega))\Bone_{\BC^{(N)}_{L_k}(0)}\|\right]\right)\\
&\leq C\left[1+\sum_{j\geq k} L_j^{\alpha s} L_j^{Nd\alpha}\left(L_j^{-2p+2Nd\alpha}+\ee^{-mL_j/2}\right)\right]
\end{align*} 
the last line in the above equation is  finite since $2p-3Nd\alpha-\alpha s\in(0,\infty) $ and $L_j=(L_0)^{\alpha^j}$ which ends the proof of Theorem \ref{thm:main.result.dynamical.loc}.

\bibliographystyle{plain}

\end{document}